\documentclass[journal,onecolumn]{IEEEtran}

\usepackage{amsmath, amssymb, amsfonts, amsthm, dsfont}
\usepackage{mathtools}
\usepackage{bm}

\usepackage{graphicx}
\usepackage{subcaption}
\usepackage{float}
\usepackage{xcolor}

\graphicspath{{figures/}}

\usepackage{booktabs}
\usepackage{array}
\usepackage{multirow}
\usepackage{makecell}
\usepackage{tabularx}
\usepackage{longtable}
\usepackage{threeparttable}

\usepackage{caption}

\usepackage{enumitem}

\theoremstyle{plain}
\newtheorem{theorem}{Theorem}
\newtheorem{lemma}{Lemma}
\newtheorem{proposition}{Proposition}

\theoremstyle{definition}

\theoremstyle{remark}
\newtheorem{remark}{Remark}

\newcommand{\Prob}{\mathbb{P}}

\newcommand{\Var}{\operatorname{Var}}

\newcommand{\argmin}{\operatorname*{arg\,min}}

\newcommand{\dd}{\,\mathrm{d}}

\newcommand{\dcs}{D_{\it CS}}

\title{Exact Sampling Under Campbell Cost via R\'enyi Channel Simulation Divergence}
\author{Yijun Fan and Aaron B.~Wagner}
\date{}

\allowdisplaybreaks

\usepackage{setspace}
\begin{document}

\maketitle

\begin{abstract}
We study the communication cost of exact sampling under Campbell's exponential coding criterion. Given a target distribution $P$ and a proposal distribution $Q$, the sender and receiver share i.i.d. $Q$-distributed samples, and the sender communicates an index so that the selected sample has distribution exactly $P$. 
We introduce a new information measure R\'enyi channel simulation divergence \(R_\alpha(P\|Q)\), a R\'enyi generalization of the channel simulation divergence. We then show that it lower bounds the Campbell cost of every exact sampler and is achievable, via a new layered sampling scheme, within an additive constant depending only on the Campbell parameter. This yields a tighter one-shot characterization than previous bounds based on R\'enyi divergences of different orders.
For i.i.d. product distributions, under suitable non-degeneracy and moment conditions, we further show that the gap beyond the leading R\'enyi-divergence term is $\frac{1}{2}\log n+O(1)$. 
\end{abstract}

\section{Introduction}

Exact sampling arises when the receiver has access to some randomness, but not to the right randomness for the task to be handled. A representative example is channel simulation. Consider a stochastic channel $P_{Y|X}$. For a realized input $x$, the desired output law is $P_{Y|X=x}$. Instead of generating such an output at the sender and communicating its value, the two parties may first share a common collection of random candidates and use communication only to decide which candidate to select. In this way, communication is used to transform shared randomness into the desired channel output, rather than to describe the output.
This idea, also known as relative entropy coding, has been studied both
theoretically \cite{harsha2007communication}\cite{li2018strong}\cite{theis2022algorithms}
and as a practical compression technique
\cite{havasi2018minimal}\cite{flamich2020compressing}\cite{flamich2022fast}.

Continuous outputs provide a particularly natural motivation for this approach. An arbitrary realization from a continuous distribution generally cannot be represented exactly by a finite binary string. Exact sampling avoids communicating the realization itself. The sender and receiver may instead share continuously valued samples, then only a discrete description of the selected candidate needs to be sent. Thus, under suitable conditions, a finite message can produce an output following an continuous distribution exactly. The formulation in this paper allows precisely this general setting: the shared randomness is drawn from a probability measure on a general measurable space, rather than from a discrete alphabet.

Exact sampling focuses on the basic task of turning shared samples from one distribution into a sample from a desired target distribution. Suppose the sender and receiver share the same collection of samples drawn from a convenient proposal distribution $Q$, while the sample they ultimately want should follow another distribution $P$. Since both parties already possess the same candidates, their values need not be communicated. What remains is to decide which candidate should be used.

At first sight, it may seem surprising that selecting among $Q$-distributed samples can produce a sample having a different distribution $P$. The key is that the candidates need not be selected uniformly. Samples that are relatively more likely under $P$ than under $Q$ should receive greater selection weight, while those that are relatively less likely under $P$ should receive less. The likelihood ratio $\dd P/\dd Q$ provides exactly this relative weighting. Classical rejection sampling gives a familiar illustration of this principle: samples are proposed according to $Q$, but accepted with probabilities determined by their relative likelihood under $P$ and $Q$, so that the accepted sample has exactly the target distribution. The Poisson functional representation \cite{li2018strong} also reweights shared proposals according to the likelihood ratio. In both cases, the shared randomness supplies the candidate values, the selection rule reshapes their frequencies, and the sender only needs to identify which shared candidate has been selected.

Recent work on exact sampling and one-shot channel simulation has focused on characterizing the communication needed to select a target distributed sample from shared randomness. Harsha et al. \cite{harsha2007communication} established expected length bounds, Li and El Gamal then \cite{li2018strong} introduced the Poisson functional representation, which has become a central tool for one-shot exact simulation. More recently, Goc and Flamich \cite{goc2024channel} introduced the channel simulation divergence to obtain a sharper characterization for causal rejection samplers, and Flamich, Sriramu, and Wagner \cite{flamich2025redundancy} used it to study the asymptotic redundancy of i.i.d.\ channel simulation.

Hill, Alajaji, and Linder~\cite{hill2026communication} recently studied exact sampling under Campbell’s exponential coding criterion. They established one-shot converse bounds involving R\'enyi divergences, obtained achievability bounds using the Poisson functional representation, and identified the optimal first-order asymptotic rate for i.i.d. product distributions. Their results show that R\'enyi divergence captures the correct first order behavior asymptotically. At the one-shot level, however, the converse and achievability bounds involve R\'enyi divergences of different orders, together with an additional optimization parameter in the achievability result. Consequently, a single quantity characterizing both sides is still missing. 

Motivated by these prior results, we seek a single one-shot quantity that both lower bounds the Campbell cost and admits a matching achievability bound. The R\'enyi channel simulation divergence $R_\alpha(P\|Q)$ introduced by this paper satisfies both.
It yields a converse for every exact sampler, while a layered Poisson construction achieves it up to an additive constant depending only on the parameter at the exponent. 
We further study the i.i.d.\ and the asymptotic behavior as the block length $n$ tends to $\infty$. Under suitable regularity conditions, the asymptotic redundancy beyond the leading R\'enyi divergence term is $\frac{1}{2}\log n + O(1)$. 


\noindent\emph{Notations.} We use $\mathbb R$ to denote the real number set, $\mathbb R_+$ to denote non-negative real number set, $\mathbb N_0$ to denote non-negative integers $\{0,1,\ldots\}$, and $\mathbb N_1$ to denote positive integers $\{1,2,\ldots\}$. Throughout this paper, $\log$ is in base $2$.

\section{Problem Formulation}

Given two probability measures $P$ and $Q$ on a measurable space $(\mathcal U,\mathcal F)$ with $P\ll Q$, we consider an exact sampling problem. Both the sender and the receiver share common randomness $\{U_i\}_{i\in\mathbb N_1}\sim Q$ i.i.d. An exact sampler outputs an index $K$ distributed over $\mathbb N_1$ such that $U_K\sim P$. Suppose the index $K$ is encoded by a uniquely decodable binary code \[
    C:\mathbb N_1\to\{0,1\}^*
\]
with codeword length $\ell(C(K))$. For $t>0$, we study the Campbell cost defined by \begin{equation}
    L(t) = \frac{1}{t}\log \mathbb E\left[2^{t\ell(C(K))}\right].
    \label{eq::def_L_t}
\end{equation}
Throughout this paper, we denote $p_k=\mathbb P(K=k)$ and let $\alpha = \frac{1}{1+t}\in (0,1)$. The R\'enyi entropy of order $\alpha$ of $K$ by \[
    H_\alpha(K) = \frac{1}{1-\alpha} \log\sum_k p_k^\alpha.
\]
Campbell's source coding theorem \cite{campbell1965coding} states that, for any uniquely decodable binary code, \begin{equation}
    H_\alpha(K)\leq L(t)
    \label{eq::lower_bd_h_alpha}
\end{equation}
and there exists a prefix code such that \begin{equation}
    L(t)<H_\alpha(K)+1,
    \label{eq::upper_bd_h_lapha}
\end{equation}
whenever $H_\alpha(K)<\infty$. Therefore, for a fixed exact sampler, the optimal Campbell cost is characterized by $H_\alpha(K)$ up to one bit. In the following sections, we characterize this quantity through the likelihood ratio between $P$ and $Q$.

As mentioned in \cite{hill2026communication}, Campbell's original proof was stated for a finite alphabet, but the same argument extends to a countably infinite alphabet. Hence, the result applies directly to the index $K\in\mathbb N_1$ considered here. Therefore, for a fixed exact sampler, the optimal Campbell coding cost is characterized by $H_\alpha(K)$ up to one bit, and the remaining problem is to characterize how small $H_\alpha(K)$ can be under the exactness constraint $U_K\sim P$. For the product problem with target $P^{\otimes n}$ and proposal $Q^{\otimes n}$, we denote by $L_n^*(t)$ the infimum of the Campbell cost over all exact samplers and uniquely decodable binary codes.

For $\beta\in(0,\infty)\setminus\{1\}$, the R\'enyi divergence of order $\beta$ between $P$ and $Q$ is defined by \cite{ervan2014renyi}
\[
D_\beta(P\|Q) := \frac{1}{\beta-1}\log \mathbb E_Q\left[\left(\frac{\dd P}{\dd Q}(X)\right)^\beta\right],
\]
whenever the right-hand side is well defined. In particular, for $\alpha\in(0,1)$,
\[
D_{1/\alpha}(P\|Q) = \frac{\alpha}{1-\alpha}\log \mathbb E_Q\left[\left(\frac{\dd P}{\dd Q}(X)\right)^{1/\alpha}\right].
\]
The R\'enyi divergence at order $\beta=1$ is defined by continuity that \[
D_1(P\|Q) := \lim_{\beta\to 1}D_\beta(P\|Q)= \int \log\frac{\dd P}{\dd Q}(u)\,\dd P(u)=D(P\|Q),
\]
where $D(P\|Q)$ is the KL divergence.
For a real-valued random variable $Z$ with density $f_Z$, its R\'enyi differential entropy of order $\alpha\in(0,1)$ is \[
h_\alpha(Z) := \frac{1}{1-\alpha}\log\int_{\mathbb R} f_Z(z)^\alpha\dd z.
\]
We will introduce a new information measure inspired by the above R\'enyi divergence. It refines the existing R\'enyi divergence based characterization of the Campbell cost and yields a tighter one-shot description.

\section{Existing Results and Contribution Summary}

\subsection{Existing Results}

Recall from \eqref{eq::def_L_t} that, whenever the corresponding limit is finite, \[
\lim_{t\downarrow0}L(t)=\mathbb E[\ell(C(K))].
\]
Thus, the expected codeword length is recovered as the limit $t\downarrow 0$ of the Campbell cost. Under this criterion, one-shot exact channel simulation has been studied through functional representations and Poisson-based constructions. Li and El Gamal~\cite{li2018strong} introduced the strong functional representation lemma and the Poisson functional representation to derive expected length achievability bounds for one-shot exact channel simulation. In the exact sampling formulation considered here, the same Poisson representation yields communication bounds in terms of $D_{1/\alpha}(P\|Q)$ \cite{hill2026communication}. 
Goc and Flamich \cite{goc2024channel} introduced the channel simulation divergence \(D_{\rm CS}(P\|Q)\) in their study of causal rejection samplers. Such samplers examine the shared \(Q\)-distributed proposals sequentially, with each accept or reject decision depending only on the proposals observed so far. Once a proposal is accepted, its index is communicated to the receiver. They showed that \(D_{\rm CS}(P\|Q)\) gives a sharper one-shot characterization than KL divergence for this class of samplers.
This channel simulation divergence can be used to characterize the second-order redundancy of non-singular i.i.d. channel simulation~\cite{flamich2025redundancy}.

The Campbell cost places increasing weight on long codewords as $t$ grows. Hence, the expected length results do not directly characterize the regime $t>0$. Hill, Alajaji, and Linder~\cite{hill2026communication} studied exact sampling under this exponential communication criterion. They derived one-shot converse bounds with leading R\'enyi-divergence terms and achievability bounds based on the Poisson functional representation.
For the product problem with target \(P^{\otimes n}\) and proposal \(Q^{\otimes n}\), let \(L_n^*(t)\) denote the infimum of the Campbell cost over all exact samplers and uniquely decodable binary codes. The work \cite{hill2026communication} also established the first-order asymptotic characterization
\[
\lim_{n\to\infty}\frac{L_n^*(t)}{n}=D_{1/\alpha}(P\|Q).
\]
These results motivate a more unified description of the Campbell cost. In this work, we pursue a one-shot quantity that can be used consistently in both converse and achievability arguments and that also supports a second order asymptotic analysis.

\subsection{Contributions}

In this paper, we introduce the R\'enyi channel simulation divergence $R_\alpha(P\|Q)$ by replacing the Shannon entropy of the width distribution in $\dcs(P\|Q)$ with its R\'enyi entropy of order $\alpha$. Equivalently, $R_\alpha(P\|Q)$ is the R\'enyi entropy of the random variable whose density is the width function $w_Q(h)$. We show that $R_\alpha(P\|Q)$ is non-increasing in $\alpha$, that its limits as $\alpha\downarrow0$ and $\alpha\uparrow1$ recover the essential supremum likelihood ratio quantity and the channel simulation divergence, respectively.
We further give an exact representation of the gap $R_\alpha(P\|Q)-D_{1/\alpha}(P\|Q)$ in terms of a R\'enyi differential entropy under an exponentially tilted measure.

For the one-shot exact sampling problem, we prove in Section~\ref{sec::converse} the converse
\[
L(t)\geq R_\alpha(P\|Q)
\]
for every exact sampler. This lower bound dominates the R\'enyi-divergence lower bounds in \cite{hill2026communication}. For achievability, we introduce a layered sampling scheme that decomposes the likelihood ratio into bounded unit-height layers. Each layer is represented by an independent marked Poisson process and is thinned according to its layer function; the first accepted point across all layers is then mapped through a bijection to a single index in the shared i.i.d. $Q$-sequence. This construction differs from the standard Poisson functional representation, which operates on a single Poisson process by rescaling arrival times with the likelihood ratio. For the layered scheme, we prove \[
L(t)<R_\alpha(P\|Q)+\frac{1}{1-\alpha}\log\frac{1}{\alpha(2^\alpha-1)}+1.
\]
Thus, for every fixed $\alpha\in(0,1)$, the same quantity $R_\alpha(P\|Q)$ appears in both the converse and achievability bounds, with an explicit additive gap depending only on $\alpha$. The resulting characterization here is therefore direct and does not require matching bounds through different R\'enyi orders.

Section~\ref{sec::asymptotic_redundancy} further sharpens the asymptotic redundancy analysis. Under non-degeneracy and moment conditions, we prove the exact logarithmic redundancy \[
R_\alpha(P^{\otimes n}\|Q^{\otimes n})-D_{1/\alpha}(P^{\otimes n}\|Q^{\otimes n})
=\frac{1}{2}\log n+O(1).
\]
This recovers the first-order asymptotic result while identifying the second-order growth of the R\'enyi channel simulation divergence. 

\section{R\'enyi Generalized Channel Simulation Divergence}
For any $u\in\mathcal U$, denote the likelihood ratio by \[
R(u) := \frac{\dd P}{\dd Q}(u).
\]
For this exact sampling problem, we measure different layers of this ratio, i.e., $\{u: R(u)\geq \ell\}$, by the \emph{width function} defined as \begin{equation}
    w_Q(\ell) := Q(\{u: R(u)\geq \ell\}), \qquad \ell\geq 0.
    \label{def::width_function}
\end{equation}
Since $R(u)\geq 0$ and $\mathbb E_Q[R]=1$, we have \[
\int_0^\infty w_Q(\ell)\dd\ell = \mathbb E_Q[R] = 1,
\]
implying that the density function $w_Q$ is a probability density on $[0,\infty)$.
In the following sections, we will show that $L(t)$ can be described by \emph{R\'enyi channel simulation divergence} of order $\alpha$, defined by \begin{equation}
    R_\alpha(P\|Q) := \frac{1}{1-\alpha}\log \int_0^\infty w_Q(\ell)^\alpha\dd \ell
    \label{def::renyi_cs_div},
\end{equation}
where $\alpha = 1/(1+t)$ is in $(0,1)$. That is, the R\'enyi entropy of order $\alpha$ of the random variable $H$ with density $w_Q(\ell)$.

Before discussing $L(t)$, we first prove some properties of $R_\alpha(P\|Q)$.
\begin{proposition}
    For probability meausres $P,Q$ on $\mathcal{U}$ with $P\ll Q$ and $\alpha\in(0,1)$, $R_\alpha(P\|Q)$ is non-increasing in $\alpha$, and \[
    \lim_{\alpha\downarrow 0}R_\alpha(P\|Q) = \log \operatorname{ess\, sup}_{u\sim Q} R(u).
    \]
    Moreover, if there exists $\alpha_0\in (0,1)$ such that $\int_0^\infty w_Q(\ell)^{\alpha_0}|\ln w_Q(\ell)|\dd\ell<\infty$,
    then \[
    \lim_{\alpha\uparrow 1} R_\alpha(P\|Q) = -\int_0^\infty w_Q(\ell)\log w_Q(\ell)\dd\ell\,=\dcs(P\|Q).
    \]
\end{proposition}
\begin{proof}
    We prove the monotonicity first. By H\"older's inequality, for any $\alpha_0, \alpha_1\in(0,1)$ and $\theta\in[0,1]$, we have \[
    \int_0^\infty w_Q(\ell)^{(1-\theta)\alpha_0+\theta \alpha_1}\dd \ell\leq \left(\int_0^\infty w_Q(\ell)^{\alpha_0}\dd\ell\right)^{1-\theta}\left(\int_0^\infty w_Q(\ell)^{\alpha_1}\dd\ell\right)^\theta,
    \]
    implying that $\log \int_0^\infty w_Q(\ell)^\alpha\dd\ell$ is convex on $(0,1)$. Define \[
    F(\alpha) = \log \int_0^\infty w_Q(\ell)^\alpha\dd\ell.
    \]
    We know that $F(1)=0$ and, by convexity, the secant slope \[
    \frac{F(1)-F(\alpha)}{1-\alpha} = \frac{-\log \int_0^\infty w_Q(\ell)^\alpha\dd\ell}{1-\alpha}
    \]
    is non-decreasing in $\alpha$ for $\alpha\in(0,1)$. Since $R_\alpha(P\|Q) = -\frac{F(1)-F(\alpha)}{1-\alpha}$, it is non-increasing in $\alpha\in(0,1)$.

    Next, consider the limit as $\alpha\downarrow 0$. We have \[
    \lim_{\alpha\downarrow 0} \int_0^\infty w_Q(\ell)^\alpha\dd\ell = \int_0^\infty \lim_{\alpha\downarrow 0} w_Q(\ell)^\alpha\dd\ell = \int_0^\infty \mathds{1}(w_Q(\ell)>0)\dd\ell = \lambda\{\ell:w_Q(\ell)>0\} = \lambda\{\ell:Q(R\geq\ell)>0\},
    \]
    which is $\operatorname{ess\, sup}_{u\sim Q} R(u)$. Since $1-\alpha$ becomes $1$ as $\alpha\downarrow 0$, we have \[
    \lim_{\alpha\downarrow 0}R_\alpha(P\|Q) = \log \operatorname{ess\, sup}_{u\sim Q}R(u).
    \]

    We then prove the limit to $1$. Since $w_Q(\ell)\in[0,1]$, for any $\alpha\in[\alpha_0,1]$, we have $w_Q(\ell)^\alpha\leq w_Q(\ell)^{\alpha_0}$, and $|w_Q(\ell)^\alpha\ln w_Q(\ell)|\leq w_Q(\ell)^{\alpha_0}|\ln w_Q(\ell)|$. Hence, the dominating convergence implies that, as $\alpha\to 1$, \[
    \int_0^\infty w_Q(\ell)^\alpha\dd\ell\rightarrow \int_0^\infty w_Q(\ell)\dd\ell = 1.
    \]  
    
    Now, we can compute the limit using L'Hospital's theorem. \begin{align*}
        \lim_{\alpha\uparrow 1}\frac{\log \int_0^\infty w_Q(\ell)^\alpha\dd\ell}{1-\alpha} &= \lim_{\alpha\uparrow 1}-\frac{\dd}{\dd \alpha} \log \int_0^\infty w_Q(\ell)^\alpha\dd\ell \\
        &= \lim_{\alpha\uparrow 1}-\frac{\left(\int_0^\infty w_Q(\ell)^\alpha\dd\ell\right)'}{\ln 2\int_0^\infty w_Q(\ell)^\alpha\dd\ell}\\
        &=\lim_{\alpha\uparrow 1}-\frac{\int_0^\infty w_Q(\ell)^\alpha\ln w_Q(\ell)\dd\ell}{\ln 2\int_0^\infty w_Q(\ell)^\alpha\dd\ell}\\
        &=-\frac{\int_0^\infty w_Q(\ell)\log w_Q(\ell)\dd\ell}{\log e\ln 2\int_0^\infty w_Q(\ell)\dd\ell} =-\int_0^\infty w_Q(\ell)\log w_Q(\ell)\dd\ell.
    \end{align*}
\end{proof}

The limit $\lim_{\alpha\uparrow 1}R_\alpha(P\|Q)$ coincides with the channel simulation divergence $\dcs(P\|Q)$ defined in \cite{goc2024channel}.

For $\alpha \in (0,1)$, let $\imath(x) := \ln R(x) = \ln\frac{\dd P}{\dd Q}(x)$ and $Q_{1/\alpha}$ denote the tilted measure \[
  \frac{\dd Q_{1/\alpha}}{\dd Q}(x) = \frac{R(x)^{1/\alpha}}{\mathbb{E}_Q[R(X)^{1/\alpha}]}.
\]

\begin{proposition}
\label{prop::renyi_width_gap}
Suppose $P\ll Q$ and $\mathbb{E}_Q[R^{1/\alpha}] < \infty$. Let $G_\alpha$ be the random variable with density \begin{equation}
\label{eq:galpha}
  g_\alpha(z) := \frac{e^{z/\alpha}\,w_Q(e^{z})}{\alpha\,\mathbb{E}_Q[R^{1/\alpha}]}, \qquad z\in\mathbb{R}.
\end{equation}
Then \begin{equation}
\label{eq:gapid}
  R_\alpha(P\|Q) - D_{1/\alpha}(P\|Q) = h_\alpha(G_\alpha) + \frac{\alpha}{1-\alpha}\log\alpha \geq 0,
\end{equation}
with \begin{equation}
  \label{eq:Grep}
  G_\alpha = \imath(X) - \alpha E,
\end{equation}
where $X\sim Q_{1/\alpha}$ and $E\sim \mathrm{Exp}(1)$ independent of $X$.
\end{proposition}
 
\begin{proof}
By the definition of $g_\alpha(z)$, we have \begin{align*}
    h_\alpha(G_\alpha) &= \frac{1}{1-\alpha}\log \int_{\mathbb R} g_\alpha(z)^\alpha\dd z\\
    &= \frac{1}{1-\alpha}\log \int_{\mathbb R} \left(\frac{e^{z/\alpha}\,w_Q(e^{z})}{\alpha\,\mathbb{E}_Q[R^{1/\alpha}]}\right)^\alpha\dd z\\
    &= \frac{1}{1-\alpha} \log \left[\frac{1}{\alpha^\alpha (\mathbb{E}_Q[R^{1/\alpha}])^\alpha}\int_{\mathbb R}e^z w_Q(e^z)^\alpha\dd z\right]\\
    &= \frac{1}{1-\alpha} \log \left[\frac{1}{\alpha^\alpha (\mathbb{E}_Q[R^{1/\alpha}])^\alpha}\int_0^\infty w_Q(h)^\alpha\dd h\right]\\
    &= \frac{1}{1-\alpha}\log \int_0^\infty w_Q(h)^\alpha\dd h - \frac{\alpha}{1-\alpha}\log \alpha - \frac{\alpha}{1-\alpha}\log\mathbb E_Q[R^{1/\alpha}]\\
    &= R_\alpha(P\|Q) - \frac{\alpha}{1-\alpha}\log \alpha - D_{1/\alpha}(P\|Q).
\end{align*}

Next, we prove that the difference $R_\alpha(P\|Q)-D_{1/\alpha}(P\|Q)$ is non-negative. For every $z\in\mathbb R$, \begin{align*}
    g_\alpha(z) &= \frac{e^{z/\alpha}Q(R\geq e^z)}{\alpha\mathbb E_Q[R^{1/\alpha}]}\\
    &= \frac{1}{\mathbb E_Q[R^{1/\alpha}]}\int \frac{1}{\alpha}e^{z/\alpha}\mathds 1(R(x)\geq e^z) \dd Q(x)\\
    &= \int \frac{1}{\alpha}e^{(z-\ln R(x))/\alpha}\mathds 1(R(x)\geq e^z) \frac{R(x)^{1/\alpha}}{\mathbb E_Q[R^{1/\alpha}]} \dd Q(x)\\
    &= \mathbb E_{Q_{1/\alpha}}\left[\frac{1}{\alpha}e^{(z-\imath(X))/\alpha}\mathds 1(\imath(X)\geq z)\right].
\end{align*}
Since $t\to t^\alpha$ is concave for $\alpha\in(0,1)$, by Jensen's inequality, \begin{align*}
    g_\alpha(z)^\alpha &= \left(\mathbb E_{Q_{1/\alpha}}\left[\frac{1}{\alpha}e^{(z-\imath(X))/\alpha}\mathds 1(\imath(X)\geq z)\right]\right)^\alpha\\
    &\geq \mathbb E_{Q_{1/\alpha}}\left[\left(\frac{1}{\alpha}e^{(z-\imath(X))/\alpha}\mathds 1(\imath(X)\geq z)\right)^\alpha\right]\\
    &= \alpha^{-\alpha}\mathbb E_{Q_{1/\alpha}}\left[e^{z-\imath(X)}\mathds 1(\imath(X)\geq z)\right].
\end{align*}
Hence, integrating over $z$, we have \begin{align*}
    \int_{\mathbb R} g_\alpha(z)^\alpha\dd z &\geq \alpha^{-\alpha}\int_{\mathbb R}\mathbb E_{Q_{1/\alpha}}\left[e^{z-\imath(X)}\mathds 1(\imath(X)\geq z)\right] \dd z\\
    &= \alpha^{-\alpha}\mathbb E_{Q_{1/\alpha}}\left[\int_{-\infty}^{\imath(X)}e^{z-\imath(X)}\dd z\right]\\
    &= \alpha^{-\alpha}\mathbb E_{Q_{1/\alpha}}\left[e^{-\imath(X)}\cdot e^{\imath(X)}\right] = \alpha^{-\alpha},
\end{align*}
implying that \[
h_\alpha(G_\alpha) \geq \frac{1}{1-\alpha} \log \alpha^{-\alpha} = -\frac{\alpha}{1-\alpha}\log \alpha.
\]

Then consider $X\sim Q_{1/\alpha}$ and $E\sim \mathrm{Exp}(1)$ with $X\perp E$. Condition on $X=x$, the conditional density of $Z|(X=x) := \imath(x)-\alpha E$ is \[
\frac{1}{\alpha}\cdot e^{\frac{z-\imath(x)}{\alpha}}\cdot\mathds 1(z\leq \imath(x)).
\]
Therefore, \begin{align*}
    g_\alpha(z) &= \frac{e^{z/\alpha}\,w_Q(e^{z})}{\alpha\,\mathbb{E}_Q[R^{1/\alpha}]}\\
    &= \frac{e^{z/\alpha}\,Q(R\geq e^z)}{\alpha\,\mathbb{E}_Q[R^{1/\alpha}]}\\
    &= \frac{e^{z/\alpha}}{\alpha\,\mathbb{E}_Q[R^{1/\alpha}]} \int \mathds 1(R(x)\geq e^z)\dd Q(x)\\
    &= \int \frac{1}{\alpha} \cdot e^{\frac{z-\imath(x)}{\alpha}} \cdot \mathds 1(R(x)\geq e^z) \cdot \frac{R(x)^{1/\alpha}}{\mathbb E_Q[R^{1/\alpha}]} \dd Q(x)\\
    &= \int \frac{1}{\alpha} \cdot e^{\frac{z-\imath(x)}{\alpha}} \cdot \mathds 1(R(x)\geq e^z) \dd Q_{1/\alpha}(x)\\
    &= f_{\imath(X)-\alpha E}(z).
\end{align*}

\end{proof}

\section{Converse}
\label{sec::converse}

The converse follows from a subadditivity property of the functional underlying the R\'enyi channel simulation divergence. For a probability space $(\Omega,\mathcal F,\mu)$ and a nonnegative measurable function $f$, define
\[
\mathcal W_\alpha(f)
:=
\int_0^\infty \mu(f\geq \ell)^\alpha\,\dd \ell,
\qquad 0<\alpha<1.
\]
This functional has exactly the same layered structure as the R\'enyi channel simulation divergence. In particular, taking $\mu=Q$ and $f=R=\dd P/\dd Q$ gives
\[
\mathcal W_\alpha(R)
=
\int_0^\infty Q(R\geq \ell)^\alpha\,\dd \ell
=
\int_0^\infty w_Q(\ell)^\alpha\,\dd \ell,
\]
implying that $R_\alpha(P\|Q) = \frac{1}{1-\alpha}\log \mathcal W_\alpha(R)$.

The usefulness of $\mathcal W_\alpha$ comes from its subadditivity. Letting $v(A):=\mu(A)^\alpha$ for $A\in\mathcal F$, we have
\[
\mathcal W_\alpha(f) = \int_0^\infty v(\{f\geq \ell\})\dd \ell,
\]
which is the Choquet integral of $f$ with respect to the capacity $v$. Since $x\mapsto x^\alpha$ is concave for $0<\alpha<1$, the capacity $v(A)=\mu(A)^\alpha$ is submodular. The subadditivity theorem for Choquet integrals with respect to submodular capacities \cite{denneberg1994subadditivity} therefore suggests that $\mathcal W_\alpha$ is subadditive.

This property particularly fits exact sampling. For any exact sampler, define \(\nu_k(A):=\Pr(K=k,U_k\in A)\) and \(r_k:=\mathrm d\nu_k/\mathrm d Q\). The likelihood ratio $R$ can be decomposed into nonnegative components associated with the possible output indices, i.e., $R=\sum_{k\geq 1} r_k$. 
Subadditivity then allows us to control $\mathcal W_\alpha(R)$ by the corresponding functionals $\mathcal W_\alpha(r_k)$. Since each component satisfies $0\leq r_k\leq 1$ and $\int r_k\,\dd Q=P(K=k)$, these terms can in turn be bounded by the $\alpha$-power sum of the distribution of $K$. This will lead directly to
\[
R_\alpha(P\|Q)\leq H_\alpha(K)\leq L(t).
\]
We first establish the required subadditivity property.

\begin{lemma}
\label{lemma::subadditivity}
    Let $(\Omega, \mathcal F, \mu)$ be a probability space. Fix $\alpha\in(0,1)$. For a non-negative measurable function $f:\Omega\rightarrow \mathbb R_+$ and countable $\{f_k\}_{k\geq 1}$ with $f_k$ non-negative and measurable, 
    we have \begin{equation}
    \mathcal W_\alpha\left(\sum_{k=1}^\infty f_k\right) \leq \sum_{k=1}^\infty \mathcal W_\alpha(f_k).
    \label{eq::subadditivity}
    \end{equation}
\end{lemma}

\begin{proof}
    Define the set function $v(A) = \mu(A)^\alpha$, $A\in\mathcal F$. We first prove that $v$ is submodular, i.e., \begin{equation}
    \label{eq::before_submodular}
    v(A) + v(B) \geq v(A\cup B) + v(A\cap B).
    \end{equation}
    Let $x = \mu(A\setminus B)$, $y=\mu(B\setminus A)$, and $z=\mu(A\cap B)$. Then \eqref{eq::before_submodular} is equivalent to \begin{equation}
    \label{eq::submodular}
    (z+x)^\alpha + (z+y)^\alpha \geq (z+x+y)^\alpha + z^\alpha.     
    \end{equation}
    For any $y,z>0$, consider the function $F(x) = (z+x)^\alpha + (z+y)^\alpha - (z+x+y)^\alpha - z^\alpha$. We then have $F(0)=0$ and $F'(x) = \alpha((z+x)^{\alpha - 1}-(z+x+y)^{\alpha - 1})\geq 0$ by $\alpha\in(0,1)$. Therefore, $F(x)\geq F(0)=0$, i.e., \eqref{eq::submodular} holds for any $x,y,z\geq 0$, implying that $v$ is submodular.

    Therefore, $\mathcal W_\alpha(f)=\int_0^\infty \mu(f\geq \ell)^\alpha\dd \ell = \int_0^\infty v(f\geq \ell)\dd\ell$ is the Choquet integral of $f$ respect to the submodular capacity $v$ and is sub-additive \cite{denneberg1994subadditivity}. Consequently, \[
    \int_0^\infty v((f+g)\geq \ell)\dd\ell \leq \int_0^\infty v(f\geq \ell)\dd\ell + \int_0^\infty v(g\geq \ell)\dd\ell
    \]
    for any measurable non-negative $f$ and $g$, which proves finite sub-additivity.

    We then use the monotone convergence theorem to prove countable sub-additivity. First, since the distribution function has at most countably many discontinuity levels, we have \[
    \mathcal W_\alpha(f) = \int_0^\infty \mu(f\geq \ell)^\alpha\dd \ell = \int_0^\infty \mu(f > \ell)^\alpha\dd \ell.
    \]
    For $S_n = \sum_{k=1}^n f_k\uparrow S$ and every $\ell > 0$, we have $\{S_n>\ell\}\uparrow\{S> \ell\}$. By continuity and monotonicity of $x^\alpha$ for $\alpha\in(0,1)$, we have $\mu(S_n>\ell)^\alpha\uparrow\mu(S>\ell)^\alpha$. Then the monotone convergence theorem gives \[
    \mathcal W_\alpha(S_n) = \int_0^\infty \mu(S_n> \ell)^\alpha \dd\ell \uparrow \int_0^\infty \mu(S> \ell)^\alpha\dd\ell = \mathcal W_\alpha(S).
    \]
    Therefore, \[
    \mathcal W_\alpha\left(\sum_{k=1}^\infty f_k\right) = \lim_{n\rightarrow \infty}\mathcal W_\alpha\left(\sum_{k=1}^n f_k\right) \leq \lim_{n\rightarrow \infty}\sum_{k=1}^n \mathcal W_\alpha(f_k) = \sum_{k=1}^\infty \mathcal W_\alpha(f_k).
    \] 
\end{proof}

Now we prove the lower bound on $L(t)$ based on this sub-additive property.
\begin{theorem}
    \label{thm::converse}
    Let $P\ll Q$. For any $t>0$ and $\alpha = 1/(1+t)\in(0,1)$, we have \[
    L(t) \geq R_{\alpha}(P\| Q).
    \]
\end{theorem}
\begin{proof}
    First, define $v_k(A) = \Prob(K=k, U_k\in A)$ for $A\in\mathcal F$. Since $U_K\sim P$, for any measurable $A\in\mathcal F$, \[
    P(A) = \Prob(U_K\in A) = \sum_{k=1}^\infty \Prob(K=k, U_k\in A) = \sum_{k=1}^\infty v_k(A).
    \]
    Hence, $P = \sum_{k=1}^\infty v_k$.
    Let $r_k(u)  = \dd v_k/\dd Q(u)$. 
    Since for any measurable $A\in\mathcal F$ \[
    v_k(A)  = \Prob(K=k, U_k\in A)\leq \Prob(U_k\in A) = Q(A),
    \]
    the ratio $0\leq r_k\leq 1$ a.e. on $Q$ and the likelihood ratio can be decomposed by \[
    R(u) = \frac{\dd P}{\dd Q}(u) = \frac{\sum_{k=1}^\infty \dd v_k}{\dd Q}(u) = \sum_{k=1}^\infty r_k(u).
    \]
    Moreover, for each $k$, we have \[
    \int r_k\dd Q = \int \dd v_k = v_k(\mathcal U) = \Prob(K=k) = p_k.
    \]

    Now apply the sub-additivity we proved in Lemma~\ref{lemma::subadditivity}. Letting \[
    \mu = Q, \qquad f = R, \qquad \text{and}\qquad f_k = r_k,
    \]
    \eqref{eq::subadditivity} gives the inequality \begin{align*}
    \int_0^\infty w_Q(\ell)^\alpha\dd\ell &= \int_0^\infty Q(R\geq \ell)^\alpha\dd\ell \\
    &= \mathcal W_\alpha(R) = \mathcal W_\alpha\left(\sum_{k=1}^\infty r_k\right) \\
    &\leq \sum_{k=1}^\infty \mathcal W_\alpha(r_k) = \sum_{k=1}^\infty\int_0^1 Q(r_k\geq \ell)^\alpha\dd\ell,
    \end{align*}
    and the concavity of $x^\alpha$ for $\alpha\in(0,1)$ and Jensen's inequality implies \[
    \int_0^1 Q(r_k\geq \ell)^\alpha \dd\ell \leq \left(\int_0^1 Q(r_k\geq \ell)\dd\ell\right)^\alpha = \left(\int r_k\dd Q\right)^\alpha = p_k^\alpha.
    \] 
    Therefore, $\int_0^\infty w_Q(\ell)^\alpha\dd\ell\leq \sum_{k=1}^\infty p_k^\alpha$. Taking $\frac{1}{1-\alpha}\log$ at both sides gives \[
    R_\alpha(P\| Q) = \frac{1}{1-\alpha}\log \int_0^\infty w_Q(\ell)^\alpha \dd \ell \leq \frac{1}{1-\alpha}\log \sum_{k=1}^\infty p_k^\alpha = H_\alpha(K).
    \]
\end{proof}

Now compare Theorem~\ref{thm::converse} with the lower bounds in \cite{hill2026communication}, where \[
L(t) \geq \max\{LB_1, LB_2\}
\]
with \begin{equation*}
\label{eq::exact_sampling_lb1}
    LB_1 = D_{1/\alpha}(P\| Q) + \frac{\alpha}{1-\alpha}\log \alpha - 1,
\end{equation*}
and \begin{equation*}
\label{eq::exact_sampling_lb2}
    LB_2 = D_{2-\alpha}(P\| Q) + \frac{1}{1-\alpha}\log \frac{1}{2-\alpha}.
\end{equation*}
Since \[
\frac{1}{\alpha} - (2-\alpha) = \frac{(1-\alpha)^2}{\alpha}\geq 0
\]
and R\'enyi divergence is non-decreasing in $\alpha$, we have \[
R_\alpha(P\| Q)\geq D_{1/\alpha}(P\| Q)\geq D_{2-\alpha}(P\| Q).
\]
Hence, Theorem~\ref{thm::converse} improves both $LB_1$ and $LB_2$ given in \cite{hill2026communication}.

\section{Achievability by Layered Sampling}
\label{sec::achievability}
Before discussing the achievability of $L(t)$, we first state the following lemma, which will be used in the later proof.
\begin{lemma}
\label{lemma::upper_bound_lt}
For $0<\alpha<1$, let $f:[0,\infty)\to[0,1]$ be non-increasing and satisfy $\int_0^\infty f(h)\dd h=1$.
Then for $m\in\mathbb N_0$, we have \[
\sum_{m\in\mathbb N_0} \left(\int_m^{m+1} f(h)\dd h\right)^\alpha \leq \frac{1}{\alpha}\int_0^\infty f(h)^\alpha\dd h.
\]
\end{lemma}

\begin{proof}
For $g:[0,1]\to[0,a]$, we first prove that \begin{equation}
\label{eq::one_block}
\left(\int_0^1 g(u)\dd u\right)^\alpha \leq \alpha\int_0^1 g(u)^\alpha\dd u + (1-\alpha)a^\alpha.
\end{equation}
If $a=0$, then $g=0$ and \eqref{eq::one_block} is immediate. Hence we assume $a>0$.

By the concavity of $x\mapsto x^\alpha$, for $x\leq a$, we have $x^\alpha-a^\alpha\leq \alpha a^{\alpha - 1}(x-a)$. Letting $x=\int_0^1 g(u)\dd u\leq a$, \[
\left(\int_0^1 g(u)\dd u\right)^\alpha \leq a^\alpha+\alpha a^{\alpha-1}\left(\int_0^1 g(u)\dd u-a\right) =
\alpha a^{\alpha-1}\int_0^1 g(u)\dd u+(1-\alpha)a^\alpha.
\]
Also, because $0\leq g(u)\leq a$ and $\alpha-1<0$, \[
g(u)^\alpha = g(u)^{\alpha - 1}\cdot g(u) \geq a^{\alpha-1}g(u).
\]
Integrating both sides gives \[
\int_0^1 g(u)^\alpha\dd u \geq a^{\alpha-1}\int_0^1 g(u)\dd u.
\]
Substituting this into the concavity bound above proves \eqref{eq::one_block}. 

Now apply \eqref{eq::one_block} to $g(u):=f(m+u)$ with $u\in[0,1]$.
Denote $a_m:=\sup_{h\in[m,m+1)} f(h)=f(m).$
In this case, \[
\int_0^1 g(u)\dd u = \int_0^1 f(m+u)\dd u = \int_m^{m+1} f(h)\dd h
\]
and, by \eqref{eq::one_block}, \begin{equation}
\label{eq::m_block}
\left(\int_m^{m+1} f(h)\dd h\right)^\alpha \leq \alpha\int_m^{m+1} f(h)^\alpha\dd h +
(1-\alpha)a_m^\alpha .
\end{equation}

Since $f$ is non-increasing, $f(h)\geq a_m$ for $h\in[m-1,m)$ a.e., implying that $\int_{m-1}^{m}f(h)\dd h \geq a_m$. 
Hence, for every $m\geq 1$, \begin{equation}
\label{eq::m_recursion}
\left(\int_m^{m+1} f(h)\dd h\right)^\alpha \leq \alpha\int_m^{m+1}f(h)^\alpha\dd h + (1-\alpha)\left(\int_{m-1}^{m}f(h)\dd h\right)^\alpha .
\end{equation}

For $m=0$, \eqref{eq::m_block} gives \begin{equation}
\label{eq::0_recursion}
\left(\int_0^1 f(h)\dd h\right)^\alpha \leq \alpha\int_0^1 f(h)^\alpha\dd h + (1-\alpha)a_0^\alpha .
\end{equation}

Summing \eqref{eq::0_recursion} and \eqref{eq::m_recursion} over $m=1,\ldots,M$, we obtain \[
\sum_{m=0}^M \left(\int_m^{m+1}f(h)\dd h\right)^\alpha \leq \alpha\int_0^{M+1}f(h)^\alpha\dd h + (1-\alpha)a_0^\alpha + (1-\alpha)\sum_{m=0}^{M-1}\left(\int_{m}^{m+1} f(h)\dd h\right)^\alpha.
\]
Moving the last sum to the left, we have \begin{equation}
\label{eq::finit_M}
\alpha\sum_{m=0}^M \left(\int_m^{m+1}f(h)\dd h\right)^\alpha + (1-\alpha)\left(\int_M^{M+1}f(h)\dd h\right)^\alpha \leq \alpha\int_0^{M+1}f(h)^\alpha\dd h + (1-\alpha)a_0^\alpha.
\end{equation}

It remains to control the initial term $a_0^\alpha$. Since $f$ is non-increasing and $a_0$ is the supremum of $f$ on the first interval $[0,1)$, we have $0\leq f\leq a_0\leq 1$ on $[0,\infty)$. Since $\alpha-1<0$, \[
f(h)^\alpha = f(h)f(h)^{\alpha-1} \geq f(h)a_0^{\alpha-1}.
\]
Therefore, \begin{equation}\label{eq::control_a_0}
\int_0^\infty f(h)^\alpha\dd h \geq a_0^{\alpha-1}\int_0^\infty f(h)\dd h = a_0^{\alpha-1}\geq a_0^\alpha,
\end{equation}
where the last inequality holds by $0<a_0\leq 1$.

Substituting \eqref{eq::control_a_0} to \eqref{eq::finit_M} and dropping the non-negative term $(1-\alpha)(\int_M^{M+1}f(h)\dd h)^\alpha$, \[
\alpha\sum_{m=0}^M \left(\int_m^{m+1}f(h)\dd h\right)^\alpha \leq \alpha\int_0^{M+1}f(h)^\alpha\dd h + (1-\alpha)\int_0^\infty f(h)^\alpha\dd h \leq \int_0^\infty f(h)^\alpha\dd h.
\]
Letting $M\to\infty$ and using monotone convergence of the partial sums gives \[
\sum_{m\in\mathbb N_0} \left(\int_m^{m+1}f(h)\dd h\right)^\alpha \leq \frac{1}{\alpha} \int_0^\infty f(h)^\alpha\dd h.
\]
\end{proof}

We now introduce a new layered sampling scheme. The main idea is to turn the likelihood ratio \(R=\mathrm dP/\mathrm d Q\) into a collection of bounded acceptance functions. We decompose \(R\) into unit-height layers \(g_m\), each satisfying \(0\leq g_m\leq 1\), so that \(g_m(x)\) can be interpreted directly as an acceptance probability. For each layer \(m\), we use an independent marked Poisson process whose marks are \(Q\)-distributed proposals. Then the superposition of all accepted processes has intensity
\[
 \sum_{m}g_m(\mathrm dx)Q(\mathrm dx)\dd t= R(\mathrm dx)Q(\mathrm dx)\dd t=P(\mathrm dx)\dd t.
\]
Thus, selecting the earliest accepted point across all layers produces an exact \(P\)-distributed sample. Moreover, the total acceptance rate of each layer is determined by the corresponding portion of the width function \(w_Q\), which will allow us to control the R\'enyi entropy of the selected index directly in terms of \(R_\alpha(P\|Q)\). The details are as follows.

Recall that $R(x)=\frac{\dd P}{\dd Q}(x)$ where $P\ll Q$ on a measurable space $\mathcal X$.
For each $m\in\mathbb N_0$, define the $m$-th function \begin{align*}
        g_m(x) &:= \int_m^{m+1}\mathds{1}(R(x)\geq h)\dd h =\begin{cases}
            0, & R(x)<m\\
            R(x) - m, &m\leq R(x)\leq m+1\\
            1, &R(x)>m+1
        \end{cases} \\
        &= \min\left\{\left(R(x) - m\right)_+, 1\right\}.
\end{align*}
Thus, $0\leq g_m(x)\leq 1$ and
\[
        \sum_{m=0}^\infty g_m(x)=\sum_{m=0}^\infty \int_m^{m+1}\mathds{1}(R(x)\geq h)\dd h = \int_0^\infty \mathds{1}(R(x)\geq h)\dd h = \int_0^{R(x)}1\dd h = R(x).
\]

Then for each $m=0,1,\ldots$, define \[
        q_m:=\int g_m(x)\dd Q.
\]
By Fubini's theorem, we have \[
        q_m = \int \int_m^{m+1} \mathds1(R(x)\geq h)\dd h\dd Q = \int_m^{m+1} \int \mathds1(R(x)\geq h)\dd Q\dd h = \int_m^{m+1}Q(R(X)\geq h)\dd h = \int_m^{m+1}w_Q(h)\dd h.
\]
Also, since $g_m(x)\geq 0$, applying monotone convergence theorem to $\sum_{m=0}^N g_m(x)$, \[
    \sum_{m=0}^\infty q_m = \int \sum_{m=0}^\infty g_m(x)\dd Q = \int R(x)\dd Q = 1.
\]

Fix a bijection $\varphi:\mathbb N_0\times\mathbb N_1\to\mathbb N_1$. For each $m\ge 0$, independently generate a rate-one Poisson process with ordered arrival times
\[
\{T_{m,n}\}_{n\in\mathbb N_1},
\]
and attach to its $n$-th point an independent mark $(X_{m,n},V_{m,n})$, where $X_{m,n}\sim Q$ and $V_{m,n}\sim \mathrm{Unif}(0,1)$.

The Poisson processes are independent across $m$, all marks ${(X_{m,n},V_{m,n})}_{m,n}$ are mutually independent, and the marks are independent of all the Poisson processes. Thus, for each fixed $m$, $\sum_{n\ge1}\delta_{(X_{m,n},T_{m,n})}$ is a marked Poisson point process on $\mathcal X\times\mathbb R_+$ with intensity $Q(\mathrm{d}x)\dd t$.

Define \[
    U_{\phi(m,n)}:=X_{m,n}.
\]
Since $\varphi$ is a bijection and the marks $X_{m,n}$ are i.i.d. according to $Q$, the sequence $(U_k)_{k\in\mathbb N_1}$ is i.i.d. $Q$. For each $m\in\mathbb N_0$, the $n$-th candidate in layer $m$ is accepted whenever \[
    V_{m,n}\le g_m(X_{m,n}).
\]

Among all accepted candidates over all $m\in\mathbb N_0$, choose the accepted
candidate with the smallest Poisson time. That is, \[
(M,N) = \argmin_{(m,n)\in\mathbb N_0\times\mathbb N_1}\{T_{m,n}: V_{m,n}\leq g_m(X_{m,n})\}.
\] 
The encoder communicates the global index
\[
        K=\varphi(M,N),
\]
and the decoder outputs
\[
        U_K=X_{M,N}.
\]

\begin{theorem}
For the layered sampling scheme above, fix $0<\alpha<1$ and assume
\[
        \int_0^\infty w_Q(h)^\alpha \dd h<\infty.
\]
Let
\[
        t=\frac{1-\alpha}{\alpha}.
\]
Then the scheme is exact, i.e., $U_K\sim P$, and there exists a prefix code for $K$ satisfying
\[
        L(t) < H_{\alpha}(K) + 1 \leq R_\alpha(P\| Q) + \frac{1}{1-\alpha} \log\frac{1}{\alpha(2^\alpha-1)} + 1.
\]
\end{theorem}

\begin{proof}
We first argue the exactness that $U_K\sim P$. 
For each layer $m$, Poisson thinning gives that the accepted points form a Poisson point process on $\mathcal X\times\mathbb R_+$ with intensity $g_m(x)Q(\dd x)\dd t$.
Since the layers are independent, their superposition is a Poisson point process with intensity \[
    \sum_{m=0}^\infty g_m(x)Q(\mathrm{d} x)\dd t = R(x)Q(\mathrm{d} x)\dd t = P(\mathrm{d} x)\dd t.
\]
Therefore, the mark of the first accepted point has distribution $P$. Since the scheme outputs exactly this mark, we obtain $U_K\sim P$, proving exactness.

Next, we compute the distribution of the selected pair $(M,N)$. For each
$m\in\mathbb N_0$, Poisson thinning gives
\begin{align*}
\text{rate of accepted proposals in layer }m
&= q_m,\\
\text{rate of rejected proposals in layer }m
&= 1-q_m,\\
\text{rate of accepted proposals from layers }j\, (j\neq m)
&= \sum_{j\neq m}q_j =1-q_m.
\end{align*}
Therefore, by Poisson superposition, the total rate of these three independent processes is $q_m+(1-q_m)+(1-q_m)=2-q_m$.
Hence, each point of the superposed process is a rejected proposal from
layer $m$ with probability $(1-q_m)/(2-q_m)$ and is an accepted proposal from layer $m$ with probability $q_m/(2-q_m)$.

The event $\{M=m,N=n\}$ occurs when the first $n-1$ points of this
superposed process are rejected proposals from layer $m$ and the
$n$-th point is an accepted proposal from layer $m$. Therefore,
\[
\Pr(M=m,N=n) =
\left(\frac{1-q_m}{2-q_m}\right)^{n-1}\cdot \frac{q_m}{2-q_m} =
\frac{q_m(1-q_m)^{n-1}}{(2-q_m)^n}.
\]
Consequently,
\[
\Pr(K=\varphi(m,n))
=
\frac{q_m(1-q_m)^{n-1}}{(2-q_m)^n}.
\]

Since $\varphi$ is a bijection, \begin{align*}
        \sum_{k\geq 1}\Prob (K=k)^\alpha &= \sum_{m\in \mathbb N_0}\sum_{n\in \mathbb N_1} \Prob (M=m,N=n)^\alpha\\
        &= \sum_m \sum_n \left(q_m \frac{(1-q_m)^{n-1}}{(2-q_m)^n}\right)^\alpha\\
        &= \sum_m q_m^\alpha \sum_n \frac{(1-q_m)^{\alpha(n-1)}}{(2-q_m)^{\alpha n}}\\
        &= \sum_m q_m^\alpha \frac{(2-q_m)^{-\alpha}}{1-\left(\frac{1-q_m}{2-q_m}\right)^\alpha}\\
        &= \sum_m q_m^\alpha \frac{1}{(2-q_m)^\alpha-(1-q_m)^\alpha}\\
        &\leq \sum_m q_m^\alpha\cdot \frac{1}{2^\alpha - 1}.
\end{align*}
The last inequality holds since for $0\le q\le1$, the function \[
        q\mapsto (2-q)^\alpha-(1-q)^\alpha
\] is non-decreasing: its derivative is \[
        \alpha\bigl((1-q)^{\alpha-1}-(2-q)^{\alpha-1}\bigr)\geq 0 .
\]
Therefore, we have $(2-q)^\alpha-(1-q)^\alpha \geq 2^\alpha-1$. 

By Lemma~\ref{lemma::upper_bound_lt}, if
$f:[0,\infty)\to[0,1]$ is non-increasing and $\int_0^\infty f(h)\dd h=1$,
then for every $0<\alpha<1$, \[
        \sum_{m=0}^\infty \left(\int_m^{m+1}f(h)\dd h\right)^\alpha \leq \frac1\alpha \int_0^\infty f(h)^\alpha\dd h .
\]
Applying this inequality to $f=w_Q$ gives \[
        \sum_{m\geq 0}q_m^\alpha \leq \frac1\alpha \int_0^\infty w_Q(h)^\alpha\dd h.
\]
Consequently, \[
        \sum_{k\geq 1}\Prob (K=k)^\alpha \leq \frac{1}{\alpha(2^\alpha-1)} \int_0^\infty w_Q(h)^\alpha\dd h.
\]
Substituting this into $L(t) < H_\alpha(K)+1$, we obtain the result \begin{align*}
L(t) < \frac{1}{1-\alpha}\log \sum_k p_k^\alpha + 1 &\leq \frac{1}{1-\alpha}\log \left(\frac{1}{\alpha(2^\alpha-1)} \int_0^\infty w_Q(h)^\alpha\dd h\right) + 1 \\
&= R_\alpha(P\| Q) + \frac{1}{1-\alpha}\log \frac{1}{\alpha(2^\alpha - 1)} + 1.
\end{align*}

\end{proof}

Combining the converse and achievability results gives an optimal characterization of $L(t)$ up to a constant for general $t>0$ \[
    R_\alpha(P\|Q) \leq L(t) < R_\alpha(P\|Q) + \frac{1}{1-\alpha}\log \frac{1}{\alpha(2^\alpha - 1)} + 1,
\]
where $\alpha = 1/(1+t)$.

\begin{remark}
    Considering the limits of this upper bound as $\alpha\uparrow 1$ and $\alpha\downarrow 0$, we have \[
        \lim_{\alpha\uparrow1}\left[\frac{1}{1-\alpha}\log\frac{1}{\alpha(2^\alpha-1)}\right] = \log e+2,
\]
 and \[
    \lim_{\alpha\downarrow 0}\left[\frac{1}{1-\alpha}\log\frac{1}{\alpha(2^\alpha-1)}\right] = \infty.
 \]

 When $\alpha\uparrow 1$, i.e., $t\downarrow 0$, the additive gap between the converse and the layered-sampling achievable bound is bounded as
 \[
 0\leq L(t)-R_\alpha(P\|Q)<\log e+3+o(1),
 \]
 where the additional one bit comes from Campbell coding. Under the integrability condition in Proposition~1, $R_\alpha(P\|Q)\to\dcs(P\|Q)$, so the bound reduces to a constant-gap characterization in the expected length $\mathbb E[\ell(C(K))]$ as limit $t\downarrow 0$, which is consistent with known results in \cite{goc2024channel}, \cite{hill2026rejection}, and \cite{hill2026communication}.


\end{remark}

\section{Asymptotic Analysis}
\label{sec::asymptotic_redundancy}

In this section we consider the extension from the one-shot exact sampling to the asymptotic problem, where $\{U_i\}_{i\geq 1}\sim Q^{\otimes n}$ i.i.d. and the selected sample follows the product distribution $P^{\otimes n}$ by communicating a single message. The previous sections have established the result \[
    \frac{R_\alpha(P^{\otimes n}\| Q^{\otimes n})}{n} \leq \frac{L_n(t)}{n} \leq \frac{R_\alpha(P^{\otimes n}\| Q^{\otimes n})}{n}  + \frac{1}{n(1-\alpha)}\log \frac{1}{\alpha(2^\alpha - 1)} + \frac{1}{n}.
\]
Therefore \[
\lim_{n\rightarrow \infty}\frac{L_n^*(t) - R_\alpha(P^{\otimes n}\| Q^{\otimes n})}{n} = 0.
\]
The rest of this section focuses on proving \[
R_\alpha(P^{\otimes n}\|Q^{\otimes n}) - D_{1/\alpha}(P^{\otimes n}\|Q^{\otimes n}) = \frac{1}{2}\log n + O(1)
\]
under certain moment conditions, implying the limit \[
\lim_{n\rightarrow \infty}\frac{L_n^*(t)}{n} = \lim_{n\to \infty}\frac{R_\alpha(P^{\otimes n}\| Q^{\otimes n})}{n} = \lim_{n\rightarrow \infty} \left(\frac{D_{1/\alpha}(P^{\otimes n}\|Q^{\otimes n})}{n} + \frac{\log n}{2n} + \frac{O(1)}{n}\right) = D_{1/\alpha}(P\| Q).
\]

Recall that, in Proposition~\ref{prop::renyi_width_gap}, we used the tilted measure $Q_{1/\alpha}$ defined by \[
  \frac{\dd Q_{1/\alpha}}{\dd Q}(x) = \frac{R(x)^{1/\alpha}}{\mathbb{E}_Q[R(X)^{1/\alpha}]}.
\]
As $\alpha$ tends to $1$, $Q_{1/\alpha}$ tends to $P$ instead of $Q$.
In the rest of this section, we write \[
  \mu_\alpha := \mathbb{E}_{Q_{1/\alpha}}[\imath(X)], \qquad V_\alpha := \mathrm{Var}_{Q_{1/\alpha}}[\imath(X)].
\]

\begin{remark}
Consider a channel \(P_{XY}\). For \(P_Y\)-a.e. \(y\) for which \(P_{X|Y=y}\ll P_X\), set $P=P_{X|Y=y}$ and $Q=P_X$.
At \(\alpha=1\), \(Q_1 = P\), and therefore \[
V_1(y) = \Var_{P_{X|Y=y}} \left[\ln\frac{dP_{X|Y=y}}{dP_X}(X)\right].
\]
Thus \(V_1(y)=0\) when the likelihood ratio is \(P_{X|Y=y}\)-a.s. constant. Also, if \(V_1(y)=0\) for \(P_Y\)-almost every \(y\), then $\frac{dP_{X|Y}}{dP_X}(X)$ is a deterministic function of \(Y\), which is the singular channel condition. 

Define \[
S_n = \sum_{i=1}^n \imath(X_i), \qquad X_1,\ldots,X_n\stackrel{\rm i.i.d.}{\sim}Q_{1/\alpha}.
\]
For the fixed pair $(P,Q)$ considered in this section, $V_\alpha>0$ plays the same non-degeneracy role under the tilted measure $Q_{1/\alpha}$. Together with the finite moment assumptions, it guarantees that the sum $S_n$ has non-degenerate fluctuations of order $\sqrt n$, which leads to the $\frac{1}{2}\log n$ term in the R\'enyi-entropy analysis below. When the corresponding variance vanishes, the log-likelihood ratio is constant on its support and this fluctuation disappears, which is consistent with the lower redundancy for singular channels.
\end{remark}

Applying Proposition~\ref{prop::renyi_width_gap} to $(P^{\otimes n},Q^{\otimes n})$ with information density $\imath_n(x^n) = \sum_{i=1}^n \imath(x_i)$ and $Q_{1/\alpha}^{\otimes n}$, directly implies
\begin{equation}
\label{eq:Gn}
  R_\alpha\left(P^{\otimes n}\|Q^{\otimes n}\right) - D_{1/\alpha}\left(P^{\otimes n}\|Q^{\otimes n}\right)
  = h_\alpha(G_{\alpha,n}) + \frac{\alpha}{1-\alpha}\log\alpha,
  \qquad
  G_{\alpha,n} = S_n - \alpha E,
\end{equation}
$\alpha E\sim \mathrm{Exp}(1/\alpha)$ is independent of $S_n$.

The next lemma bounds R\'enyi entropy of order $\alpha$ from two sides: how tall the density is bounds the width from below, and how fast the density decays bounds it from above.

\begin{lemma}
\label{lem:bounds}
Let $0 < \alpha < 1$. Suppose a random variable $U$ has density $f$. Then the following hold:
\begin{enumerate}
\item[(i)] $h_\alpha(U) \geq -\log \|f\|_\infty$.
\item[(ii)] If $\Prob(|U| > t) \le C t^{-\gamma}$ for some $C>0$, $\gamma > \frac{1-\alpha}{\alpha}$, and every $t \geq 1$, then \[
  h_\alpha(U) \leq \frac{1}{1-\alpha}\log\left(2^{1-\alpha}\left(1 + \frac{C^\alpha}{1-2^{1-\alpha-\alpha\gamma}}\right)
  \right).
\]
\end{enumerate}
\end{lemma}
 
\begin{proof}
Both parts use $\int f = 1$ and H\"older's inequality: for any measurable $A \subseteq \mathbb{R}$, \begin{equation}
\label{eq:holderA}
  \int_A f^{\alpha}\dd u \leq |A|^{1-\alpha}\left(\int_A f\dd u\right)^{\alpha}.
\end{equation}
 
(i) Since $\alpha - 1 < 0$ and the density $f$ is non-negative, \[
\int f^{\alpha}\dd u = \int f\cdot f^{\alpha-1} \dd u \geq \|f\|_\infty^{\alpha-1}\int f\dd u=\|f\|_\infty^{\alpha - 1}.
\]
Then applying $\frac{1}{1-\alpha}\log(\cdot)$ proves (i).
 
(ii) Split $\mathbb{R}$ into $\{|u| \le 1\}$ and the $A_j := \{2^{j} < |u| \leq 2^{j+1}\}$,
$j \geq 0$. On $\{|u|\leq 1\}$, \eqref{eq:holderA} gives $\int f^{\alpha} \leq 2^{1-\alpha}$. On $A_j$, we have $|A_j| = 2^{j+1}$ and $\int_{A_j} f\dd u \leq \Prob(|U| > 2^{j}) \leq C 2^{-j\gamma}$ by assumption.
Therefore, \[
  \int_{A_j} f^{\alpha} \leq \left(2^{j+1}\right)^{1-\alpha}\left(C2^{-j\gamma}\right)^{\alpha} = 2^{1-\alpha}C^\alpha 2^{j(1-\alpha-\alpha\gamma)}.
\]
Since $\gamma > \frac{1-\alpha}{\alpha}$ the exponent $1-\alpha-\alpha\gamma$ is negative, we have \[
\sum_{j\geq 0} \int_{A_j} f^\alpha \dd u \leq 2^{1-\alpha}C^\alpha \sum_{j\geq 0} (2^{1-\alpha-\alpha\gamma})^j = \frac{2^{1-\alpha}C^\alpha}{1-2^{1-\alpha-\alpha\gamma}}.
\]
Adding the central piece gives the stated bound in (ii).
\end{proof}

\begin{theorem}
\label{thm:asymp}
Suppose $P \ll Q$, $V_\alpha>0$, and \[
\mathbb{E}_Q\left[|\ln R|^{p}\cdot R^{1/\alpha}\right] < \infty \quad\text{for some } p \geq 2 \text{ and } p > \frac{1-\alpha}{\alpha}.
\] Then we have \begin{equation}
\label{eq::asymptotic_main}
  R_\alpha\left(P^{\otimes n}\big\|Q^{\otimes n}\right) - D_{1/\alpha}\left(P^{\otimes n}\big\|Q^{\otimes n}\right) = \frac{1}{2}\log n + O(1),
\end{equation}
\end{theorem}
 
\begin{proof}
Let $X_1, X_2, \dots$ are i.i.d.\ $Q_{1/\alpha}$ and $S_n:= \sum_{i=1}^{n}\imath(X_i)$, as the setting in \eqref{eq:Gn}.
 
We first prove that $\mathbb E_Q[R^{1/\alpha}]<\infty$ and $\mathbb E_{Q_{1/\alpha}}|\imath|^p <\infty$ from the assumption. Observe that $R^{1/\alpha}\leq e^{1/\alpha}$ on $R\leq e$ and $R^{1/\alpha}\leq |\ln R|^p\cdot R^{1/\alpha}$ on $R>e$, we have \[
\mathbb E_Q[R^{1/\alpha}] \leq e^{1/\alpha} + \mathbb E_Q[|\ln R|^p\cdot R^{1/\alpha}] < \infty.
\]
Also, by Jensen's inequality, we have $\mathbb E_Q[R^{1/\alpha}]\geq (\mathbb E_Q[R])^{1/\alpha}=1$. Then \begin{align}
\label{eq:allmoments}
    \mathbb E_{Q_{1/\alpha}}|\imath|^p = \int |\imath|^p \dd Q_{1/\alpha} = \int |\imath|^p\cdot\frac{R^{1/\alpha}}{\mathbb E_Q[R^{1/\alpha}]}\dd Q \leq \mathbb E_Q[|\ln R|^p\cdot R^{1/\alpha}] < \infty.
\end{align}

Since $p\geq 2$, H\"older's inequality gives $\mathbb E_{Q_{1/\alpha}}[\imath^2] \leq \left(\mathbb E_{Q_{1/\alpha}}[|\imath|^p]\right)^{2/p} <\infty$. 
Consequently, by Cauchy--Schwarz, $|\mu_\alpha|=\left|\mathbb E_{Q_{1/\alpha}}[\imath]\right|\leq \left(\mathbb E_{Q_{1/\alpha}}[\imath^2]\right)^{1/2}<\infty$, implying that \[
    V_\alpha = \Var_{Q_{1/\alpha}}(\imath) = \mathbb E_{Q_{1/\alpha}}[\imath^2]-\mu_\alpha^2 <\infty.
\]

Moreover, for the centered $p$-th moment, using $|a-b|^p\le 2^{p-1}(|a|^p+|b|^p)$, we obtain  \begin{align*}
    \mathbb E_{Q_{1/\alpha}}\left[|\imath - \mu_\alpha|^p\right] \leq 2^{p-1}\left(\mathbb E_{Q_{1/\alpha}}\left[|\imath|^p+|\mu_\alpha|^p\right]\right) \leq 2^{p-1}\left(\mathbb E_{Q_{1/\alpha}}\left[|\imath|^p\right]+\mathbb E_{Q_{1/\alpha}}\left[|\mu_\alpha|^p\right]\right)<\infty.
\end{align*}

By \eqref{eq:Gn} it now suffices to prove $h_\alpha(G_{\alpha,n}) = \frac{1}{2}\log n + O(1)$.
R\'enyi differential entropy satisfies $h_\alpha(aU+b) = h_\alpha(U) + \log a$ for $a > 0$. Hence, with \[
  \widetilde G_n := \frac{G_{\alpha,n} - n\mu_\alpha}{\sqrt n}
  = \frac{S_n - n\mu_\alpha}{\sqrt n} - \frac{\alpha E}{\sqrt n},
\]
we have \begin{equation}
\label{eq:scaling}
  h_\alpha(G_{\alpha,n}) = \frac{1}{2}\log n + h_\alpha\left(\widetilde G_n\right),
\end{equation}
so the theorem is equivalent to proving $h_\alpha(\widetilde G_n)=O(1)$. 
 
Next, we lower bound $h_\alpha(\widetilde G_n)$ by bounding its density \[
\widetilde g_n(u) = \sqrt{n}g_n(n\mu_\alpha+\sqrt{n}u)
\] 
with $\|\widetilde g_n\|_\infty=\sqrt{n}\|g_n\|_\infty$ using concentration function of $S_n$. Conditioning on $S_n$ and using that $\alpha E$ has the non-increasing density $\frac{1}{\alpha} e^{-v/\alpha}$ on $[0,\infty)$, \begin{align*}
  g_n(z) &= \mathbb{E}\left[\frac{1}{\alpha} e^{-(S_n-z)/\alpha}\mathds{1}(S_n \geq z)\right]\\
  &=\sum_{j\geq 0}\mathbb E\left[\frac{1}{\alpha}e^{-(S_n-z)/\alpha}\mathds 1(S_n-z\in[j,j+1)])\right]\\
  &\leq \sum_{j\geq 0}\mathbb E\left[\frac{1}{\alpha}e^{-j/\alpha}\mathds 1(S_n-z\in[j,j+1)])\right]\\
  &= \sum_{j\geq 0}\frac{1}{\alpha}e^{-j/\alpha}\Prob \left(S_n - z \in [j,j+1)\right)\\
  &\leq \sup_{y\in\mathbb{R}}\Prob \left(S_n \in [y,y+1)\right)\cdot \frac{1}{\alpha}\sum_{j\geq 0}e^{-j/\alpha}\\
  &= \frac{1}{\alpha(1-e^{-1/\alpha})}\cdot\sup_{y\in\mathbb{R}}\Prob \left(S_n \in [y,y+1)\right).
\end{align*}
 
It remains to bound that mass. Since $V_\alpha > 0$, we can choose $a < b$ with $\Prob(\imath \leq a) > 0$ and $\Prob(\imath \geq b) > 0$, and an $\ell\in (0,\min\{1, b-a\})$. Then any interval of length $\ell$ is disjoint from at least one of $(-\infty,a]$ and $[b,\infty)$.
Hence, \[
  q_\ell := \sup_{y\in\mathbb{R}}\Prob \left(\imath \in [y,y+\ell]\right) \leq  1 - \min\left\{\Prob(\imath \leq a),\,\Prob(\imath \geq b)\right\} < 1 .
\]
The Kolmogorov--Rogozin concentration inequality \cite{rogozin1961increase} applied to the $n$ i.i.d.\ summands of $S_n$ with common interval length $\ell \leq 1$ gives \begin{align}
    \sup_{y\in\mathbb{R}}\Prob \left(S_n \in [y,y+1)\right) &\leq \sup_{y\in\mathbb{R}}\Prob \left(S_n \in [y,y+1]\right)\notag\\
    &\leq \lceil 1/\ell\rceil\cdot \sup_{y\in\mathbb{R}}\Prob \left(S_n \in [y,y+\ell]\right)\notag\\
    &\leq \frac{\lceil 1/\ell\rceil\cdot C_R}{\left(\sum_{k=1}^n \left(1-\sup_{y\in\mathbb R}\Prob \left(\imath(X_k)\in[y,y+\ell]\right)\right)\right)^{1/2}}\notag\\
    &=\frac{\lceil 1/\ell\rceil\cdot C_R}{\left(n\left(1-\sup_{y\in\mathbb R}\Prob \left(\imath \in[y,y+\ell]\right)\right)\right)^{1/2}}\notag\\
    &= \frac{\lceil 1/\ell\rceil\cdot C_R}{\left(n(1-q_\ell)\right)^{1/2}} = \frac{1}{\sqrt{n}}\cdot \frac{\lceil 1/\ell\rceil\cdot C_R}{\sqrt{1-q_\ell}}. \label{eq:rogozin}
\end{align}
 
Combining, $\|g_n\|_\infty \le C_1 n^{-1/2}$ with $C_1=\frac{\lceil 1/\ell\rceil C_R}{\sqrt{1-q_\ell}\cdot \alpha(1-e^{-1/\alpha})}$. The density of $\widetilde G_n$ therefore satisfies $\|\widetilde g_n\|_\infty\leq C_1$. By Lemma~\ref{lem:bounds} (i), we have \begin{equation}
\label{eq:lower}
  h_\alpha\left(\widetilde G_n\right) \geq -\log \|\widetilde g_n\|_\infty \geq -\log C_1.
\end{equation}
 
Then we show that $\widetilde G_n$ obeys a tail bound $\Prob(|\widetilde G_n| > t) \le C t^{-\gamma}$ for all $t \ge 1$, uniformly in $n$, with some
$\gamma > \frac{1-\alpha}{\alpha}$, so that Lemma~\ref{lem:bounds} (ii) applies. Since the summands of $S_n - n\mu_\alpha$ are i.i.d.\ and centered, we have \[
\mathbb{E}\left[\frac{S_n - n\mu_\alpha}{\sqrt n}\right] = 0, \qquad  \mathbb{E}\left[\left(\frac{S_n - n\mu_\alpha}{\sqrt n}\right)^2\right] = \frac{\mathrm{Var}(S_n)}{n} = V_\alpha \in (0,\infty),
\] 
and, by Rosenthal's inequality \cite{rosenthal1970subspaces} for sums of i.i.d.\ centered random variables with $p > 2$, \begin{align*}
  \mathbb{E}\left|S_n - n\mu_\alpha\right|^p &= \mathbb E\left|\sum_{k=1}^n\left(\imath(X_k)-\mu_\alpha\right)\right|^p\\
  &\leq C_p\left(\sum_{k=1}^n\mathbb E|\imath(X_k)-\mu_\alpha|^p + \left(\sum_{k=1}^n \mathbb E[(\imath(X_k)-\mu_\alpha)^2]\right)^{p/2}\right)\\
  &=C_p\left(n\mathbb E_{Q_{1/\alpha}}|\imath -\mu_\alpha|^p + \left(nV_\alpha\right)^{p/2}\right)<\infty.
\end{align*}
Since $\mathbb{E}[E^{p}] = \Gamma(p+1)$, Minkowski's inequality in $L^p$ gives
\begin{align*}
    \left\|\widetilde G_n\right\|_{p} &\leq \left\|\frac{|S_n-n\mu_\alpha|}{\sqrt{n}}\right\|_{p} + \frac{\alpha}{\sqrt n}\|E\|_{p}\\
    &\leq C_p (n^{1-p/2}\mathbb E_{Q_{1/\alpha}}|\imath -\mu_\alpha|^p + \left(V_\alpha\right)^{p/2})^{1/p} + \alpha\,\Gamma(p+1)^{1/p} \\
    &\leq C_p (\mathbb E_{Q_{1/\alpha}}|\imath -\mu_\alpha|^p + \left(V_\alpha\right)^{p/2})^{1/p} + \alpha\,\Gamma(p+1)^{1/p} =: C_2,
\end{align*}
and Markov's inequality gives $\Prob(|\widetilde G_n| > t) \leq C_2^p t^{-p}$ with $\gamma = p>\frac{1-\alpha}{\alpha}$.

For $p=2$, we have \[
\mathbb E\left[\left|\frac{S_n-n\mu_\alpha}{\sqrt n} \right|^2\right]=
\frac1n\mathbb E\left[\left|\sum_{i=1}^n(\imath(X_i)-\mu_\alpha)\right|^2\right] =
\frac1n\sum_{i=1}^n\mathbb E\left[|\imath(X_i)-\mu_\alpha|^2\right] =V_\alpha.
\]
Hence, \[
\mathbb E|\widetilde G_n|^2 =
\mathbb E\left|
\frac{S_n-n\mu_\alpha}{\sqrt n}
\right|^2 + \frac{\alpha^2}{n}\mathbb E[E^2]
= V_\alpha+\frac{2\alpha^2}{n} 
\leq V_\alpha+2\alpha^2.
\]
Then Chebyshev's inequality gives $\Pr(|\widetilde G_n|>t)\leq \frac{V_\alpha+2\alpha^2}{t^2}$. Then we can take $\gamma=2$, satisfying the condition in Lemma~\ref{lem:bounds} (ii) since we assumed $2=p>\frac{1-\alpha}{\alpha}$.

Hence, Lemma~\ref{lem:bounds} (ii) applies with the stated $\gamma$ and $C$, implying that \begin{equation}
\label{eq:upper}
  h_\alpha\left(\widetilde G_n\right) \leq \frac{1}{1-\alpha}\log\left[2^{1-\alpha}\left(1 + \frac{C^{\alpha}}{1-2^{\,1-\alpha-\alpha\gamma}}\right)\right].
\end{equation}
 
By \eqref{eq:lower} and \eqref{eq:upper} the sequence $h_\alpha(\widetilde G_n)$ is bounded for any $n\geq 1$ by constants depending only on $\alpha$ and $\imath$ under $Q_{1/\alpha}$. Substituting into \eqref{eq:scaling}
gives $h_\alpha(G_{\alpha, n}) = \frac{1}{2}\log n + O(1)$, and \eqref{eq:Gn} then gives \eqref{eq::asymptotic_main}.
\end{proof}
 
\begin{remark}
The threshold \(\alpha=1/3\) determines whether a variance constraint alone controls the Rényi differential entropy from above. When \(\alpha>1/3\), the maximum R\'enyi entropy result under a variance constraint \cite{johnson2007some} gives directly
\[
h_\alpha(G_{\alpha,n})
\le
\frac12\log(nV_\alpha+\alpha^2)+C_\alpha
=
\frac12\log n+O(1).
\]
Thus, in this regime, finite nonzero variance is sufficient for the upper bound. When \(\alpha\le1/3\), finite variance alone does not control \(h_\alpha\). This is why Theorem~\ref{thm:asymp} imposes the stronger \(p\)-th moment condition and uses the polynomial tail arguments.
\end{remark}

\clearpage
\bibliographystyle{IEEEtran}
\bibliography{ref.bib}

\end{document}